\documentclass[letterpaper]{article} 
\usepackage{aaai25}  
\usepackage{times}  
\usepackage{helvet}  
\usepackage{courier}  
\usepackage[hyphens]{url}  
\usepackage{graphicx} 
\usepackage{natbib}  
\usepackage{caption} 
\usepackage{algorithm}
\usepackage{algorithmic}

\usepackage{newfloat}
\usepackage{listings}
\DeclareCaptionStyle{ruled}{labelfont=normalfont,labelsep=colon,strut=off} 
\floatstyle{ruled}
\newfloat{listing}{tb}{lst}{}
\floatname{listing}{Listing}
\usepackage[utf8]{inputenc} 
\usepackage{url}            
\usepackage{booktabs}       
\usepackage{amsfonts}       
\usepackage{nicefrac}       
\usepackage{microtype}      
\usepackage{xcolor}         

\usepackage{amsmath,bm}
\usepackage{amsthm}
\usepackage{amssymb}
\usepackage{mathtools}
\usepackage{centernot}
\usepackage{algorithm}
\usepackage{algorithmic}
\usepackage[algo2e]{algorithm2e} 
\usepackage{comment}

\DeclareMathOperator*{\argmax}{arg\,max}

\usepackage{tikz}
\usepackage{tikz-cd}

\newcommand{\E}{\operatorname{\mathbb E}}

\newtheorem{corollary}{Corollary}
\newtheorem{theorem}{Theorem}
\newtheorem{proposition}{Proposition}
\newtheorem{definition}{Definition}
\newtheorem{assumption}{Assumption}

\newtheorem{altassumption}{Assumption}[assumption]
\newenvironment{assumption+}[1]
  {\renewcommand{\thealtassumption}{\ref*{#1}$'$}%
   \begin{altassumption}}
  {\end{altassumption}}

\newtheorem{subassumption}{Assumption}[assumption]

\newenvironment{assumption*}
 {\ifnum \value{subassumption}=0
 \stepcounter{assumption}\fi\subassumption}
 {\endsubassumption}
 
\newcommand{\bck}[1]{\left\langle{#1}\right\rangle}
\newcommand\scal[2]{\bck{{#1},{#2}}}

\usepackage{cleveref}
\crefalias{subassumption}{assumption}
\crefalias{altassumption}{assumption}
\crefname{assumption}{assumption}{Assumptions}
\crefname{appendix}{appendix}{Appendices}

\newcommand{\BE}{\text{Q}^{*}\text{RE}}
\newcommand{\QRE}{\text{Q}^{\pi}\text{RE}}

\allowdisplaybreaks

\usepackage[textsize=tiny]{todonotes}
\setuptodonotes{inline}

\title{Bounded Rationality Equilibrium Learning in Mean Field Games

}

\author{%
  Yannick Eich, Christian Fabian, Kai Cui, Heinz Koeppl
}
\affiliations{
    Dept. of Electrical Engineering and Information Technology, Technische Universität Darmstadt\\

    \{yannick.eich, heinz.koeppl\}@tu-darmstadt.de
}

\begin{document}

\maketitle

\begin{abstract}
Mean field games (MFGs) tractably model behavior in large agent populations. The literature on learning MFG equilibria typically focuses on finding Nash equilibria (NE), which assume perfectly rational agents and are hence implausible in many realistic situations.
To overcome these limitations, we incorporate bounded rationality into MFGs by leveraging the well-known concept of quantal response equilibria (QRE). Two novel types of MFG QRE enable the modeling of large agent populations where individuals only noisily estimate the true objective. We also introduce a second source of bounded rationality to MFGs by restricting the agents' planning horizon. The resulting novel receding horizon (RH) MFGs are combined with QRE and existing approaches to model different aspects of bounded rationality in MFGs.
We formally define MFG QRE and RH MFGs and compare them to existing equilibrium concepts such as entropy-regularized NE. Subsequently, we design generalized fixed-point iteration and fictitious play algorithms to learn QRE and RH equilibria. After a theoretical analysis, we give different examples to evaluate the capabilities of our learning algorithms and outline practical differences between the equilibrium concepts.
\end{abstract}


\section{Introduction}
Learning equilibria in multi-agent games is of great practical interest but hard to scale to many agents \citep{daskalakis2009complexity, deng2023complexity}. Mean field games (MFGs) allow scaling to arbitrarily many exchangeable agents at fixed complexity. MFGs are of recent interest as a tractable method to learn approximate equilibria of rational, selfish agents \citep{guo2019learning, cui2021approximately, xie2021learning, lauriere2022learning, anahtarci2020q}. Thus, MFGs are applied in various settings ranging from finance to engineering \citep{djehiche2017mean, achdou2020mean, carmona2020applications}.

A common concept in multi-agent learning is the Nash equilibrium (NE), where each player's strategy is optimal given others', leading to no incentive for agents to change strategies.
The optimality notion inherent in NE assumes full rationality of the individual agents.

However, in many real-world situations individuals may not behave perfectly rational due to limited information processing capabilities, psychological factors, social considerations or other factors. Deviations from perfect rationality are described by the fundamental concept of bounded rationality \citep{simon1955behavioral, simon1979rational, kahneman1982psychology, selten1990bounded, gigerenzer2002bounded, kahneman2013perspective}. Bounded rationality implies that for many real-world scenarios NE are insufficient due to their rigorous perfect rationality assumption. Instead of NE, we require a more realistic equilibrium concept accounting for partially irrational agents. 

A popular game-theoretic approach to modeling bounded rationality of agents are quantal response equilibria (QRE) \citep{MCKELVEY19956,mckelvey1998quantal} which are used, e.g., in economics \citep{breitmoser2010understanding}, robust RL \citep{reddi2023robust} and for efficient NE approximation \citep{gemp2023approximating}. Intuitively, in a QRE agents perceive rewards perturbed by noise and act optimally with respect to these perturbed rewards.
In our work, we extend QRE to the domain of MFGs to model the behavior of a large number of agents who deviate from perfect rationality.

Meanwhile, on the control-theoretic side, a common approximately optimal control method is model predictive control (MPC) \citep{kouvaritakis2016model}, also known as receding horizon control. To further enhance modeling of bounded rationality in MFGs, we incorporate a receding horizon method, where agents make decisions based on a limited future time horizon, reflecting more realistic decision-making processes. In contrast to MPC-based variants of MFGs such as \cite{inoue2021model}, we analyze the resulting novel receding horizon equilibria and instead focus on \textit{learning} such equilibria, in a \textit{discrete-time} setting.

Beyond realism, introducing bounded rationality yields possible tractability advantages. NE computation for MFGs can be hard, motivating the search for alternative equilibrium notions. We show that under certain assumptions, QRE can be computed using a fixed-point iteration (FPI). Moreover, QRE solutions can be seen as NE approximations with arbitrarily accurate design \citep{eibelshauser2019markov}. Recently, different equilibria have been introduced as NE approximations in MFGs \citep{cui2021approximately}. We compare QRE with these equilibria theoretically and empirically and provide a new algorithm to compute QRE which extends to these equilibria. 
For receding horizon equilibria, we develop novel algorithms effective in theory and practice.

Our main contributions are:
\begin{itemize}
    \item We formulate QRE for MFGs to incorporate bounded rationality for a more realistic MFG framework;
    \item We integrate a receding horizon method tailored to the limited lookahead capacity of realistic agents;
    \item We give theoretical and empirical results to put MFG QRE in context to existing equilibrium concepts;
    \item We generalize the known fictitious play (FP) and FPI algorithms for NE to learn QRE and other equilibria;
    \item We provide empirical examples to demonstrate the capabilities of our learning algorithms.
\end{itemize}

\section{Equilibria in MFGs} \label{sec:discrete}
In this section, we first give a description of finite games in discrete time and their corresponding MFGs. We then define common and new equilibrium notions as solution concepts and desired results of multi-agent equilibrium learning algorithms, which are introduced thereafter. For space reasons, some proofs are in the appendix.

\textit{Notation: Denote by $\mathcal P(\mathcal X)$ the space of probability measures on finite set $\mathcal X$, equipped with the $L_1$ norm $\lVert \cdot \rVert$ unless noted otherwise. Equip products of metric spaces with the $\sup$ metric. Further, let $[N] \coloneqq \{ 1, \ldots, N \}$ for $N \in \mathbb N$.} 

\subsection{Finite Agent Games}
For the finite $N$-agent game of practical interest, consider agents $i \in [N]$ endowed with random states $x^i_t$ and actions $u^i_t$ at all times $t \in \mathcal T \coloneqq \{ 0, 1, \ldots, T-1\}$ up to time horizon $T \in \mathbb N$. Let $\mathcal X$ and $\mathcal U$ be the finite state and action spaces for agents, respectively. The empirical mean field (MF) $\mu_t^N \coloneqq \frac{1}{N} \sum_{i=1}^N \mathbf 1_{x_t^i}$ can be understood as a histogram of agent states. Each agent $i$ implements stochastic Markovian policies $\pi^i \in \Pi \equiv \mathcal P(\mathcal U)^{\mathcal X \times \mathcal T}$ depending on the current time and local agent state. For some initial state distribution $\mu_0$ with $x^i_0 \sim \mu_0$, for all agents $i$ define state-action dynamics
\begin{align} \label{eq:finite}
    u^i_t &\sim \pi^i_t(u^i_t \mid x^i_t), \quad
    x^i_{t+1} \sim p_t(x^i_{t+1} \mid x^i_t, u^i_t, \mu_t^N)
\end{align}
given some transition kernels $p_t \colon \mathcal X \times \mathcal U \times \mathcal P(\mathcal X) \to \mathcal P(\mathcal X)$.

Competitive agents aim to optimize their own objective while predicting other agents' behavior. The objective notion depends on the chosen equilibrium notion
, but typically uses functions $r_t \colon \mathcal X \times \mathcal U \times \mathcal P(\mathcal X) \to \mathbb R$ resulting in rewards $r_t(x^i_t, u^i_t, \mu_t^N)$ at all times $t \in \mathcal T$ to maximize.

\subsection{Mean Field Games}
MFGs are the limit of finite $N$-agent games with $N \to \infty$ and approximate many-agent finite games well. By a law of large numbers, the empirical MF is essentially replaced by its deterministic limiting MF. 
The idea of MFGs is to find approximate (symmetric) equilibria, which are otherwise hard to find in finite games with many agents \citep{deng2023complexity}.
MFGs assume all agents to symmetrically play the same policy $\pi^* \in \Pi$ -- the equilibrium solution.  
Whenever agent $i$ deviates from $\pi^*$ and instead uses some policy $\pi$, this corresponds to the policy tuple $(\pi, \underline{\pi}^{-i})$, where $\underline{\pi}^{-i} = (\pi^*, \ldots, \pi^*)$ denotes all but the $i$-th policy in the finite game. 
Hence, in the limit as $N \to \infty$ we have
\begin{align} \label{eq:mf}
    u_t &\sim \pi_t(u_t \mid x_t), \quad
    x_{t+1} \sim p_t(x_{t+1} \mid x_t, u_t, \mu_t)
\end{align}
for the representative deviating agent. Here, the empirical MFs $\mu^N_t$ are replaced by the deterministic limiting MF $\mu \coloneqq (\mu_t)_{t \in \mathcal T} \in \mathcal M \subseteq \mathcal P(\mathcal X)^{\mathcal T}$, given by the probability law $\mu$ of any other agent playing the assumed equilibrium policy $\pi^*$. Further, $\mathcal M$ is the space of all obtainable MFs. We write $\mu = \Gamma_{\mathcal M}(\pi^*)$, defined by fixed initial $\mu_0$ and the recursion
\begin{align*}
    \mu_{t+1}(x') = \sum_{x \in \mathcal X} \mu_t(x) \sum_{u \in \mathcal U} \pi^*_t(u \mid x) p_t(x' \mid x, u, \mu_t).
\end{align*}

\subsection{Notions of Non-Cooperative Equilibria} \label{sec:notions}
As discussed, there are many equilibrium notions.
Here, we focus on non-cooperative equilibria where agents optimize over independent policies to maximize their own objective.

\subsubsection{Nash equilibria.}
First, we have the standard objective of any agent $i$ given as
\begin{align*}
    J^i(\hat \pi, \underline{\pi}^{-i}) = \E \left[ \sum_{t \in \mathcal T} r(x^i_t, u^i_t, \mu^N_t) \right],
\end{align*}
which, for $\underline{\pi}^{-i} = \times_{j \neq i} \pi$, in the limiting MFG yields
\begin{align} \label{eq:J}
    J(\hat \pi, \pi) &= \E \left[ \sum_{t \in \mathcal T} r(x_t, u_t, \mu_t) \right], \quad \mu = \Gamma_{\mathcal M}(\pi)
\end{align}
where only agent $i$ deviates from policy $\pi$ to $\hat \pi$. 
If agents are rational and anticipate other agents' decisions, all agents should use policies such that none can improve their objective by deviating from the equilibrium. This leads to the well-known Nash equilibrium.

\begin{definition}[Exploitability]
     In the finite game, $\mathcal E^N(\underline{\pi}) \coloneqq \max_{i \in [N]} \sup_{\hat \pi \in \Pi} \{ J^i(\hat \pi, \underline{\pi}^{-i}) - J^i(\underline{\pi}) \}$ is the exploitability of a policy tuple $\underline{\pi} \in \Pi^N$. The limiting exploitability $\mathcal E(\pi)$ of a policy $\pi \in \Pi$ is $\mathcal E(\pi) \coloneqq \max_{\hat \pi \in \Pi} J(\hat \pi, \pi) - J(\pi, \pi)$.
\end{definition}

\begin{definition}[Approximate NE]
    For any $\epsilon > 0$, an $\epsilon$-approximate NE ($\epsilon$-NE) is a policy tuple $\underline{\pi} \in \Pi^N$ with $\mathcal E^N(\underline{\pi}) \leq \epsilon$. An exact NE is an $\epsilon$-NE with $\epsilon = 0$.
\end{definition}

The resulting limiting NE thus becomes a policy that performs optimally against itself, i.e. when all other agents also use the same policy \citep{saldi2018markov}.

\begin{definition}[Mean Field NE]\label{def:NE}
    A Mean Field NE (MFNE) is a policy $\pi^* \in \Pi$ such that $\pi^* \in \argmax_{\pi \in \Pi} J(\pi, \pi^*)$.
\end{definition}

MFNE are well-known to constitute approximate NE in large finite $N$-agent games, in the sense of a negligible exploitability, rigorously motivating MFGs and MFNE under mild continuity assumptions of the game.

\begin{assumption} \label{ass:cont}
    The transition kernel $P$ and reward function $r$ are continuous in their MF argument.
\end{assumption}

\begin{proposition}[{\citet[Thm. 3.3, 4.1]{saldi2018markov}}]
    Under Assm.~\ref{ass:cont}, a MFNE $\pi^*$ exists, and yields a finite game $\epsilon$-NE $\underline{\pi}^* = (\pi^*, \ldots, \pi^*)$, with $\epsilon \to 0$ as $N \to \infty$. 
\end{proposition}

The maximization of Eq.~\eqref{eq:J} for a fixed MF $\mu$ involves the optimal state-action value function given by the Bellman recursion $Q^{*} \equiv Q^{\mu, *} = \Gamma_{Q^*}(\mu)$ (suppressing $\mu$) defined as
\begin{multline}\label{eq:Qstar}
    Q^{*}_t(x,u) = r(x,u,\mu_t) + \sum_{x' \in \mathcal{X}} p_t(x'\mid x,u,\mu_t) \\
    \max_{u' \in \mathcal U} Q^{*}_{t+1}(x', u'),
\end{multline}
with $Q_{T-1}^{*}(x,u) = r (x, u, \mu_{T-1})$. The optimal policy is then obtained by maximizing $Q^*$ with respect to action $u$, for which we write $\pi^*=\Gamma_{\Pi}^{*}(Q^*)$. We can then rewrite Def.~\ref{def:NE} as the fixed-point equation $\pi^*= \Gamma_{\Pi}^{*}(\Gamma_{Q^*}(\Gamma_{\mathcal M}(\pi^*)))$.

\subsubsection{Regularized equilibria.}
A common alternative to MFNE is to use regularized control \citep{geist2019theory, belousov2019entropic} (typically entropy regularization). 
The idea is to replace the objective in Eq.~\eqref{eq:J} by an entropy-regularized one, which maximizes the entropy $\mathcal H(\hat \pi_t(\cdot \mid x_t)) \coloneqq - \sum_{u \in \mathcal U} \hat \pi_t(u \mid x_t) \log \hat \pi_t(u \mid x_t)$ of policies in encountered states $x_t$,
\begin{align*}
    J^{\mathrm{RE}}_\alpha(\hat \pi, \pi) \coloneqq \E \left[ \sum_{t \in \mathcal T} r(x_t, u_t, \mu_t) + \alpha \mathcal H(\hat \pi_t(\cdot \mid x_t)) \right]
\end{align*}
with $\mu = \Gamma_{\mathcal M}(\pi)$, and temperature $\alpha > 0$. Accordingly, we define regularized equilibria (RE) similar to MFNE.
\begin{definition}[RE] \label{def:RE}
    A RE is a policy $\pi^* \in \Pi$ with $\pi^* = \argmax_{\pi \in \Pi} J^{\mathrm{RE}}_\alpha(\pi, \pi^*)$.
\end{definition}

Note that for fixed $\mu \in \mathcal{M}$, it is known (e.g., \citet{cui2021approximately}) that the optimal policy $\hat \pi^{\mu, \alpha}$ is
\begin{align}
    \hat{\pi}_{t}^{\mu, \alpha}(u\mid x) = \frac{\exp \left(\frac{1}{\alpha} \tilde{Q}_{t}^{\mu, \alpha}(x,u)\right)}{\sum_{u' \in \mathcal{U}} \exp \left(\frac{1}{\alpha} \tilde{Q}_{t}^{\mu, \alpha}(x,u')\right)}
\end{align}
and write $\hat \pi^{\mu,\alpha} = \Gamma^\alpha_\Pi(\tilde{Q}^{\mu, \alpha})$ for such softmax policies, given the soft state-action value function $\tilde{Q}_{t}^{\mu, \alpha}(x,u)$. We also write $\tilde{Q}_{t}^{\mu, \alpha} = \Gamma^\alpha_{\tilde Q}(\mu)$ for the soft state-action value function given $\mu$, defined through the smooth-maximum Bellman recursion
\begin{multline} \label{eq:smoothQ}
    \tilde{Q}_{t}^{\mu, \alpha}(x,u) = r (x, u, \mu_t) + \sum_{x' \in \mathcal{X}} p_t (x' \mid x, u, \mu_t) \\
    \cdot\alpha \log \left( \sum_{u' \in \mathcal U} \exp \left(\frac{1}{\alpha} \tilde{Q}_{t+1}^{\mu, \alpha}(x',u')\right) \right)
\end{multline}
with $\tilde{Q}_{T-1}^{\mu, \alpha}(x,u) = r (x, u, \mu_{T-1})$. For readability, we omit super- and subscript $\mu, \alpha$ where it is clear from context. Now, we can rewrite Def.~\ref{def:RE} as $\pi^* = \Gamma^\alpha_\Pi(\Gamma^\alpha_{\tilde Q}(\Gamma_{\mathcal M}(\pi^*)))$.

There are various reasons for regularization, such as improved tractability in learning MFG equilibria \citep{anahtarci2020q, cui2021approximately, lauriere2022learning, li2024transition}, exploratory properties in RL settings \citep{guo2022entropy}, and robustness against model uncertainty \citep{eysenbach2021maximum}. 

\subsubsection{Quantal response equilibria.}
To incorporate bounded rationality, we assume agents act suboptimally and do not exactly optimize an objective.
Building on economics literature \citep{breitmoser2010understanding, eibelshauser2019markov}, we introduce Markov QRE
as an MFG equilibrium notion where agents only noisily estimate the state-action value function. Depending on the noise, agents act independently according to their own estimates, to the best of their knowledge.
In the mentioned literature, the QRE definition is based on the state-action values $Q^{\pi} \equiv Q^{\mu, \pi} = \Gamma_{Q^{\pi}}(\mu, \pi)$ of a policy $\pi$ under current MF $\mu$, given by the recursion
\begin{multline*}
    Q^{\pi}_t(x,u) = r(x,u,\mu_t) + \sum_{x' \in \mathcal{X}} p_t(x'\mid x,u,\mu_t) \\
    \sum_{u' \in \mathcal U} \pi_t(u' \mid x') Q^{\pi}_{t+1}(x', u'),
\end{multline*}
with $Q_{T-1}^{\pi}(x,u) = r (x, u, \mu_{T-1})$. We extend this to the optimal state-action value function $Q^{*}$ defined in Eq.~\eqref{eq:Qstar} and denote the resulting equilibria as  $\text{Q}^{\pi}\text{RE}$ and $\text{Q}^{*}\text{RE}$, respectively. 
For $\text{Q}^{\pi}\text{RE}$, the noisy state-action value function is
\begin{equation*}
    \hat Q^{\pi}_t(x,u) = Q^{\pi}_t(x,u) + \epsilon_t(x,u),
\end{equation*}
given policy $\pi$, where $\epsilon_t$ is sampled from a distribution $p(\epsilon_t)$. 
We then describe the set of realizations of $\epsilon_t$, where agents in state $x$ perceive action $u$ as optimal, called response set, as
\begin{align*}
    \mathcal{R}_{t,x,u} = \{\epsilon_t
    \colon \hat Q^{\pi}_t(x,u) > \hat Q^{\pi}_t(x,u') \quad \forall u'\neq u \}.
\end{align*}
The probability that agents choose action $u$ corresponds to the probability mass of the respective response set. Next, we define the corresponding equilibrium.

\begin{definition}[$\text{Q}^{\pi}\text{RE}$]
    A (Markov) $\text{Q}^{\pi}\text{RE}$ is a policy $\pi^*$ s.t.
    \begin{equation*}
        \pi_t^{*}(u\mid x) = \int_{\mathcal{R}_{t,x,u}}p(\epsilon_t) \mathrm d\epsilon_t
    \end{equation*}
    for all $t \in \mathcal T, x \in \mathcal X, u \in \mathcal U$, and $\mu = \Gamma_{\mathcal M}(\pi^*)$.
\end{definition}
For the $\text{Q}^{*}\text{RE}$ we define the noisy estimates of $Q^*$ as
\begin{equation*}
    \hat Q^{*}_t(x,u) = Q^{*}_t(x,u) + \epsilon_t(x,u).
\end{equation*}
We equivalently define the resulting response set $\mathcal{R}^*_{t,x,u}$ and the corresponding equilibrium.
\begin{definition}[$\text{Q}^*$RE]
    A (Markov) $\text{Q}^*$RE is a policy $\pi^*$ s.t.
    \begin{equation*}
        \pi_t^{*}(u\mid x) = \int_{\mathcal{R}^*_{t,x,u}}p(\epsilon_t) \mathrm d\epsilon_t
    \end{equation*}
    for all $t \in \mathcal T, x \in \mathcal X, u \in \mathcal U$, and $\mu = \Gamma_{\mathcal M}(\pi^*)$.
\end{definition}
If the noise follows a Gumbel distribution with parameter $\lambda$, $\text{Q}^{\pi}\text{RE}$ and $\text{Q}^{*}\text{RE}$ policies can be computed analytically for fixed $\mu$, as softmax policies $\pi^{*} = \Gamma^{1/\lambda}_\Pi({Q}^{\pi})$ and $\pi^{*} = \Gamma^{1/\lambda}_\Pi({Q}^{*})$, respectively, where $\alpha=1/\lambda$.
The special case of $\text{Q}^{\pi}\text{RE}$ leads to the MFG analogue of the so-called logit equilibrium, considered in economics \citep{breitmoser2010understanding, eibelshauser2019markov}, while the special case of $\text{Q}^{*}\text{RE}$ leads to the Boltzmann equilibrium \citep{guo2019learning, cui2021approximately}.

\begin{definition}[Logit $\text{Q}^{\pi}\text{RE}$]
    A Logit $\text{Q}^{\pi}\text{RE}$ (L$\text{Q}^{\pi}\text{RE}$) is a policy $\pi^* \in \Pi$ such that $\pi^* = \Gamma^{1/\lambda}_\Pi(\Gamma_{Q^{\pi}}(\Gamma_{\mathcal M}(\pi^*),\pi^*))$.
\end{definition}
\begin{definition}[Boltzmann equilibrium]
    A BE is a policy $\pi^* \in \Pi$ such that $\pi^* = \Gamma^{1/\lambda}_\Pi(\Gamma_{Q^*}(\Gamma_{\mathcal M}(\pi^*)))$.
\end{definition}

In the following, $\text{Q}^{\pi}\text{RE}$ and $\text{Q}^{*}\text{RE}$ usually refer to their special cases of L$\text{Q}^{\pi}\text{RE}$ and BE. We show that such equilibria are guaranteed to exist. For BE see \citep{cui2021approximately}.

\begin{proposition} \label{prop:disc-qre-exist}
    For any $\lambda > 0$, a $\text{Q}^{\pi}\text{RE}$ exists under Assm.~\ref{ass:cont}.
\end{proposition}

\subsubsection{Other MFG equilibrium concepts.}

The literature contains many equilibrium concepts for MFGs. One example are (coarse) correlated equilibria \citep{campi2022correlated, muller2022learning, muller2021learning} where agents obtain advice from a mediator to align their actions. The resulting correlation mechanism enables efficient equilibria calculation under moderate assumptions compared to NE. Furthermore, there are Stackelberg equilibria for MFGs \citep{elie2019tale, carmona2021finite, carmona2022mean, vasal2022master} where one principal tries to optimally incentivize a mean field of infinitely many agents. Since a detailed discussion of these and many more MFG equilibrium concepts is beyond the scope of our paper, we focus on NE, $\QRE$, $\BE$, and RE instead and refer to the above references and \citet{fudenberg1991game} for a general, not MFG specific overview of equilibrium concepts.

\subsubsection{Receding horizon equilibria.}
Next, we describe a second method to model bounded rationality that can be combined with the previously defined equilibrium concepts. We introduce receding horizon (RH) equilibria to model limited lookahead capacity of agents by considering a shorter horizon for the underlying Bellman recursion. They describe the behaviour of agents with a model predictive control (MPC) \cite{kouvaritakis2016model}, where decisions are based on a shortened future horizon, and therefore allows for more realistic or practical MFG models.

The previously defined objectives $J$, e.g. for NE or RE, are sums over the whole time horizon $\mathcal T$. In the RH scenario, however, agents plan ahead only the next $H$ time steps beyond the current time $t\in \mathcal T$.
Like in MPC, we assume that agents apply the first action of the resulting policy and then repeat the optimization for the next time step.
A RH equilibrium is thus an ensemble of sequential MFG equilibria.

For RH MFG, define the respective RH NE as follows.
For an agent at time $t$, the RH objective given the MF policy $\pi$ is
        \begin{align*}
         J^H_{t} (\hat\pi, \pi)
        \coloneqq  \E \left[ \sum_{t' = t}^{\min (T, t+H)} r \left(x_{t'}, u_{t'}, \mu_{t, t'}^H \right) \right],
    \end{align*}
with $\mu_{t}^H = \Gamma_{\mathcal M, t}^H \left(\pi\right)$, 
defined by initial $\mu_{t,t}^H$ and the recursion
\begin{align*}
    \mu_{t, t'+1}^H (x') = \sum_{x \in \mathcal X} \mu_{t, t'}^H(x) \sum_{u \in \mathcal U} \pi_{t'}(u \mid x) p_{t'}(x' \mid x, u, \mu_{t, t'}^H ).
\end{align*}
\begin{definition}[RH NE]
    For a horizon $H \in \mathbb N$, a RH NE is a policy ensemble $(\pi^{*,H}_{t})_{t \in \mathcal T} \in \Pi^{T}$ such that for all $t \in \mathcal T$
    \begin{align*}
        \pi^{*, H}_{t} \in \argmax_{\pi \in \Pi} J^H_{t} (\pi, \pi^{*, H} _{t}),\textrm{ with} \, \mu_{t}^H = \Gamma_{\mathcal M, t}^H \left(\pi^{*, H}_{t} \right),
    \end{align*}
where the initial MF for each MFG is the MF of the previous MFG after one time step, i.e. $\mu_{t, t}^H = \mu_{t-1, t}^H$ for all $t>0$ and $\mu_{0, 0}^H = \mu_{0}$.
$(J_{t}^H)_{t \in \mathcal T}$ is the corresponding RH objective ensemble. Since MFs may deviate in practice and the horizon moves forward by one after each time step $t$, agents only implement the first entry $\pi^{*, H}_{t, t}$ of each policy. Thus, the implemented equilibrium policy $\pi^{**,H} \in \Pi$ results from the policy ensemble $(\pi^{*, H}_{t})_{t \in \mathcal T} \in \Pi^{T}$ by taking
    \begin{align*}
        \pi^{**, H}_{t} = \pi^{*, H}_{t, t},
    \end{align*}
    for each $t \in \mathcal{T}$, i.e., the diagonal $\pi^{**,H} = \mathrm{diag} ((\pi^{*, H}_{t})_{t \in \mathcal T})$.
\end{definition}

The RH concept extends to regularized equilibria by changing the corresponding objective. Accordingly, we define an approximate $\epsilon$-RH RE such that for any $t$, $\pi^{*, H}_{t}$ is $\epsilon$-optimal instead of exactly maximizing $\pi \mapsto J^H_{t} (\pi, \pi^{*, H} _{t})$, and we define their exploitability as the sum of exploitabilities in each sub-MFG at time $t$. In a similar fashion we define the RH QRE in the appendix.

\subsection{Connections between Equilibrium Notions} \label{sec:eq-explain}

\begin{figure}
    \centering
    \includegraphics[width=0.99\linewidth]{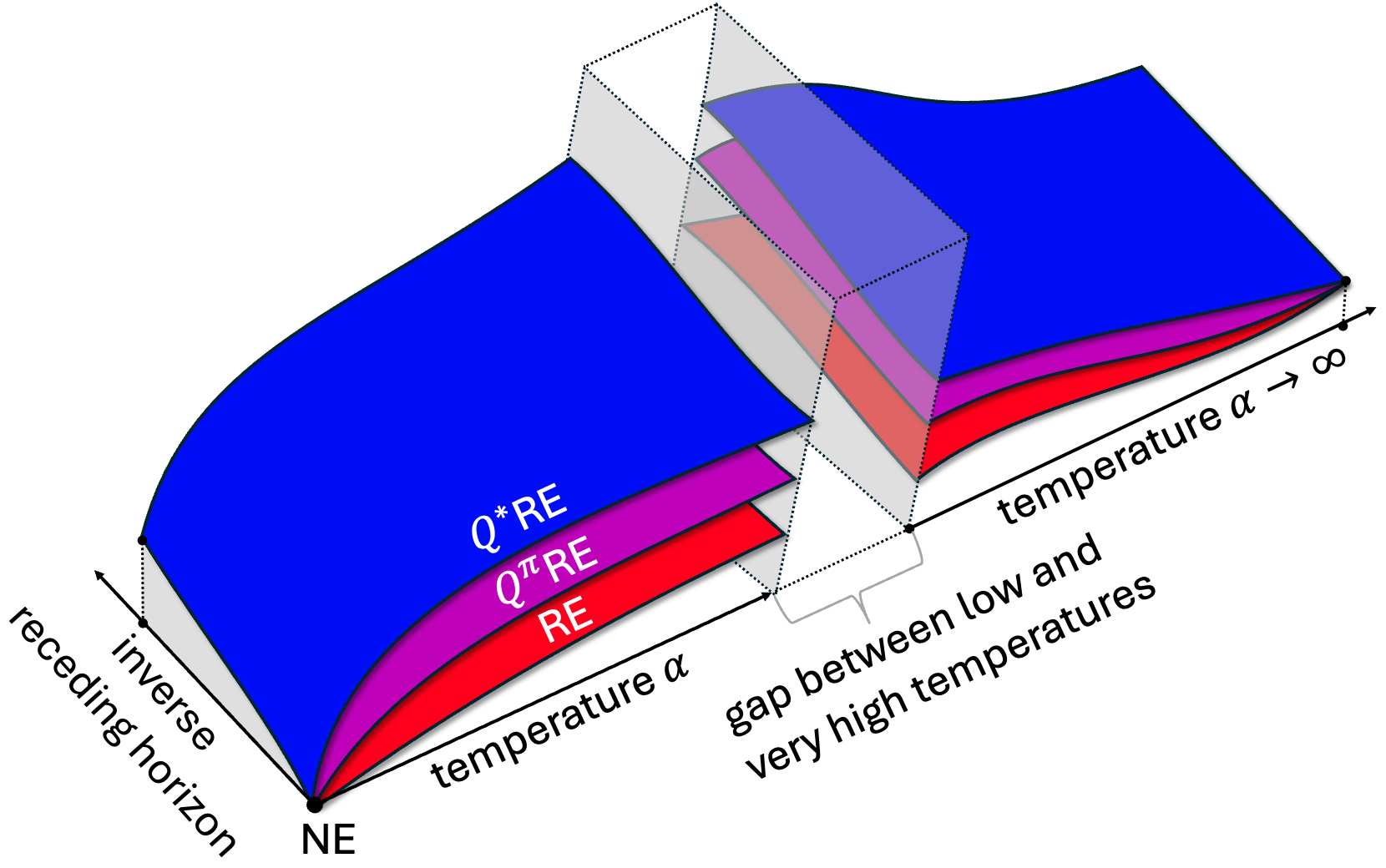}
    \caption{A visualization of the resulting equilibrium policies (one-dimensional for illustration) of QRE and RE over temperature and receding horizon. In the limit of low temperature and infinite horizon, all concepts become the MFNE. In the limit of infinite temperature, all solutions become the constant uniform policy.}
    \label{fig:vis_overview}
\end{figure}  

In the following, we compare the notions of equilibria introduced in the prequel. In all mentioned equilibria, $\pi^*$ can be written down using the equivalent fixed-point equations
\begin{subequations} \label{eq:equilibria}
\begin{align}
    \text{(NE)} \quad \pi^* &= \Gamma_{\Pi}^{*}(\Gamma_{Q^*}(\Gamma_{\mathcal M}(\pi^*))),\\
    (\QRE) \quad \pi^* &= \Gamma^{1/\lambda}_\Pi(\Gamma_{Q^{\pi}}(\Gamma_{\mathcal M}(\pi^*), \pi^*)), \\
    (\BE) \quad \pi^* &= \Gamma^{1/\lambda}_{\Pi}(\Gamma_{Q^*}(\Gamma_{\mathcal M}(\pi^*))), \\
    \text{(RE)} \quad \pi^* &= \Gamma^\alpha_{\Pi}(\Gamma^\alpha_{\tilde Q}(\Gamma_{\mathcal M}(\pi^*))).
\end{align}
\end{subequations}

The decomposed definition enables a comparison between different equilibria. If we choose $\alpha = 1/\lambda$, the $\QRE$, $\BE$ and RE policies are all of $\alpha$-softmax form. Therefore, although the equilibria have various differing derivations, their special cases considered here are connected.

\paragraph{Distinctiveness of equilibrium concepts.}
In general, the equilibrium notions are distinct. Here, we look at the special cases of entropy-regularized RE and $\QRE$ / $\BE$ with Gumbel noise (LQRE / BE) where the difference between $\QRE$, $\BE$ and RE is the usage of \textit{policy}, \textit{optimal} and \textit{soft} state-action value functions respectively. As $\alpha \to 0$, RE are known to essentially become NE in the sense of exploitability in the finite system \citep{cui2021approximately}. Furthermore, as $1/\lambda \to 0$, the Gumbel noise vanishes in $\QRE$ / $\BE$. Thus, in the low temperature limit, all equilibrium notions are equivalent which is visualized in Figure~\ref{fig:vis_overview}. On the other hand, as $\alpha = 1/\lambda \to \infty$, by Assm.~\ref{ass:cont}, each solution tends to the uniform policy. However, for intermediate temperatures $\alpha$, the equilibria are distinct, see Figure~\ref{fig:vis_overview}. Analogously, for RH equilibria we have convergence to standard non-RH equilibria as the horizon becomes large, $H \to \infty$.

\paragraph{$\QRE$ is a first order approximation of RE.}
Although $\QRE$ and RE are distinct, we find a connection between both through a principled approximation. Indeed, let $\alpha = 1 / \lambda$, then the $\QRE$ is obtained by solving a recursive first order approximation of the soft state-action value function.
\begin{theorem} \label{thm:first-order-approx}
    $\QRE$ are obtained from RE by first-order approximation of the smooth-maximum Bellman equation \eqref{eq:smoothQ}.
\end{theorem}

\begin{proof}[Proof of Thm.~\ref{thm:first-order-approx}]
Assume w.l.o.g. that $\vert \mathcal{U} \vert = n$ for some finite $n \in \mathbb N$ and define the function $g \colon \mathbb{R}^n \to \mathbb{R}$, $g(x_1, \ldots, x_n) = \alpha \log \left( \sum_{i =1}^n \exp \left(x_i / \alpha  \right) \right)$
such that
\begin{align*}
    \frac{\partial}{\partial x_j} g (x_1, \ldots, x_n) = \frac{\exp \left(x_j / \alpha\right)}{\sum_{i \leq n} \exp \left(x_i / \alpha\right)} \, .
\end{align*}

Then, the smooth-maximum Bellman recursion in the context of regularized equilibria can be rewritten as 
\begin{multline*}
   \tilde{Q}_{t}^{\mu, \alpha}(x,u) = r (x, u, \mu_t) \\+ \sum_{x' \in \mathcal{X}} p_t (x' \mid x, u, \mu_t) \cdot g \left( \tilde{Q}_{t+1, \alpha}^{\mu}(x', \cdot)  \right) \, .
\end{multline*}
Similarly, the Bellman equation in the $\QRE$ setup with Gumbel noise and $\lambda = 1/ \alpha$ can be expressed as
\begin{multline*}
    Q^{\mu, *}_t(x,u) = r(x,u,\mu_t) \\
    + \sum_{x' \in \mathcal{X}} p_t(x'\mid x,u,\mu_t) \scal{\nabla g( Q^{\mu, *}_{t+1}(x', \cdot))}{Q^{\mu, *}_{t+1}(x', \cdot)}
\end{multline*}
by using $\pi^*_t ( \cdot \vert x) = \nabla g( Q^{\mu, *}_{t}(x, \cdot))$, and $g(\mathbf 0) = \alpha \log n$. Thus, we interpret the $\QRE$ Bellman equation as a first order approximation of the smooth-maximum Bellman recursion, with added error term $\alpha \log n \to 0$ as $\alpha \to 0$. Here, we recursively estimate $\tilde Q$ through $Q^*$ backwards in $t\in \mathcal T$.
\end{proof}
Thm.~\ref{thm:first-order-approx} establishes a rigorous connection between $\QRE$ and RE, in the common case of Gumbel noise and entropy regularization. While both concepts share similarities, our empirical evaluations demonstrate that they can yield different results in general, see the experiments section.

\section{Learning Non-Cooperative MF Equilibria}\label{sec:algorithms}
There is a variety of methods for computing or learning NE in MFGs, each with its own limitations. Standard fixed-point iteration (FPI) is not guaranteed to converge, as it cannot be Lipschitz even in simple standard finite MFGs \citep[Thm.~2]{cui2021approximately}. Meanwhile, other algorithms such as fictitious play (FP) and Online Mirror Descent require monotonicity and their theory is currently limited to dynamics independent of the mean field \citep{perrin2020fictitious, perolat2021scaling}. Other recent ideas include an optimization-based approach \citep{guo2024mf} and a regularization approach \citep{guo2019learning, cui2021approximately}, for which convergence via FPI is guaranteed for strong enough regularization. In the following, we generalize FPI and FP to our equilibrium concepts of interest.

\begin{algorithm}
    \caption{Generalized Fixed-Point Iteration (GFPI).}
    \label{alg1}
    \begin{algorithmic}[1]
        \STATE Input: Temperature $\alpha = 1/\lambda>0$, initial policy $\pi^0$, equilibrium type ET $\in \{\text{NE}, \QRE, \BE, \text{RE}\}$.     
        \STATE Define $\Gamma_{\Pi}$ and $\Gamma_{Q}$ according to ET (Eq.~\ref{eq:equilibria}).
        \FOR {$k=0,1,\ldots$,K-1}
         \STATE Evaluate $\pi^{k+1} \leftarrow \Gamma_{\Pi}(\Gamma_{Q}(\Gamma_{\mathcal M}(\pi^k),\pi^k))$.
        \ENDFOR
        \STATE \textbf{return} $\pi^K$
    \end{algorithmic}
\end{algorithm}

\paragraph{Fixed point iteration.}
In the FPI approach, one repeatedly computes the result of the right-hand side of the fixed point equations \eqref{eq:equilibria}. 
Start with some initial policy, e.g., $\pi^{0}_t(u \mid x) = 1 / |\mathcal U|$ for all $t,x,u$. Then, for each iteration $k=0,1,\ldots$, compute the resulting MF by solving the fixed point equation, the resulting value functions under the new MF, and finally the new policy, e.g. for RE as
\begin{align*}
    \pi^{k+1} = \Gamma^\alpha_\Pi(\Gamma^\alpha_{\tilde Q}(\Gamma_{\mathcal M}(\pi^k))).
\end{align*}
For NE, $\QRE$ and $\BE$ the operators are changed according to the desired setting as in Eqs.~\eqref{eq:equilibria}. The overall generalized FPI (GFPI) algorithm is given in Alg.~\ref{alg1} and converges for sufficiently high temperatures under Lipschitz conditions. 

\begin{assumption} \label{ass:Lcont}
    The transition kernel $P$ and reward function $r$ are Lipschitz continuous in their MF argument.
\end{assumption}

\begin{proposition} \label{prop:disc-conv-fpi}
    In MFGs, under Assm.~\ref{ass:Lcont}, FPI converges to a $\QRE$ / $\BE$ / RE  for sufficiently high $\alpha > 0$.
\end{proposition}

Such a result is known \citep{cui2021approximately} for RE and BE ($\BE$) and extended towards $\QRE$ here.
As a result, the convergence of FPI follows for high temperatures. 

\begin{corollary}
    By Banach's fixed-point theorem, regularized FPI converges to an equilibrium for sufficiently large $\alpha > 0$.
\end{corollary}

Requiring sufficiently large $\alpha$ limits applicability of the GFPI algorithm.
Therefore, alternate or more general learning algorithms are desired. 

\begin{algorithm}[b]
    \caption{Generalized Fictitious Play (GFP).}
    \label{alg2}
    \begin{algorithmic}[1]
        \STATE Input: Temperature $\alpha > 0$, policy $\pi^0$, $\beta \in (0, 1)$, equilibrium type ET $\in \{\text{NE}, \QRE, \BE, \text{RE}\}$.     
        \STATE Define $\Gamma_{\Pi}$ and $\Gamma_{Q}$ according to ET (Eq.~\ref{eq:equilibria}).
        \STATE Initialize $\mu^0\leftarrow\Gamma_{\mathcal M}(\pi^0)$ as the MF induced by $\pi^0$.
        \STATE Initialize weighted sum of policies $\bar \pi^0 = \pi^0 \mu^0$
        \FOR {$k=0,1,\ldots,K-1$}
        \STATE Evaluate $\pi^{k+1} \leftarrow \Gamma_{\Pi}(\Gamma_{Q}(\mu^k,\pi^k))$.
        \STATE Compute MF $\mu^{k+1}\leftarrow\Gamma_{\mathcal{M}}(\pi^{k+1})$ induced by $\pi^{k+1}$.
        \STATE Compute $\bar\pi^{k+1}  = \beta\bar\pi^{k+1}+(1-\beta) \mu^{k+1} \pi^{k+1}$ .
        \STATE Average MF $\mu^{k+1} \leftarrow (1-\beta)\mu^{k+1} + \beta \mu^k $.
        \STATE Normalize $\pi^{k+1} \propto \bar\pi^{k+1}$.
        \ENDFOR
        \STATE \textbf{return} $\pi^K$.
    \end{algorithmic}
\end{algorithm}
\paragraph{Fictitious play.}
One such algorithm for standard NE is the FP algorithm. Parallel to GFPI, we formulate the generalized FP (GFP) algorithm in Alg.~\ref{alg2} for learning MFG equilibria. For NE, the algorithm matches with the proven FP algorithm in \citet{perrin2020fictitious}, while for $\mathrm{BE}$, RE the algorithm matches with initial experiments in \citet{cui2021approximately}. Thus, GFP is known to converge for NE under certain assumptions \citep{perrin2020fictitious}.
We extend the FP convergence proof in standard MFGs towards RE, and also towards RH equilibria by proposing suitable modifications, see the red plane in Fig.~\ref{fig:vis_overview}. As a side result, we make the existing proof for FP convergence more precise by explicitly verifying the usage of the envelope theorem, which is strictly speaking only applicable for any non-zero regularization, $\alpha > 0$. The proofs for the remaining cases are left to future work. The proof requires a standard monotonicity assumption on the considered MFG, as well as a common simplifying assumption on the system dynamics and is based on a continuous-time ODE version of the algorithm with time-continuous iterations as in \citet{perrin2020fictitious}. Then, the actual algorithm can be interpreted as a discretization of the continuous-time version.
\begin{assumption}[\citet{lasry2007mean}] \label{ass:monotone}
    The game is monotone, i.e. $R_t (x, u, \mu) = r_t (x, \mu) +\bar{r}_t (x, u)$ with differentiable $r_t$ and $\sum_{x \in \mathcal{X}} \left( \mu (x) - \mu' (x) \right) \left( r_t (x, \mu) - r_t (x, \mu')\right) \leq 0$ for all $t \in \mathcal T$ and all $\mu, \mu' \in \mathcal P(\mathcal X)$.
\end{assumption}
\begin{assumption}[\citet{perrin2020fictitious, perolat2021scaling}] \label{ass:pnomf}
    The transition kernel does not depend on the MF.
\end{assumption}
\begin{theorem} \label{thm:gfp1}
    Under Assm.~\ref{ass:monotone} and \ref{ass:pnomf}, the continuous-time version of GFP for RE converges to zero exploitability, $\mathcal{E}^{\mathrm{RE}}(\bar \pi^\tau) \coloneqq \max_{\pi'} J^{\mathrm{RE}}_\alpha(\pi', \bar \pi^\tau) - J^{\mathrm{RE}}_\alpha(\bar \pi^\tau, \bar \pi^\tau) \to 0$ at rate $\mathcal O(\frac 1 \tau)$, with algorithm run time $\tau$.
\end{theorem}

\begin{algorithm}[b!]
    \caption{Sequential RH-GFP.}
    \label{alg3}
    \begin{algorithmic}[1]
        \STATE Input: Temperature $\alpha > 0$, policy $\pi^0$, $\beta \in (0, 1)$, equilibrium type ET $\in \{\text{NE}, \QRE, \BE, \text{RE}\}$, receding horizon $H$.     
        \STATE Define $\Gamma_{\Pi}$ and $\Gamma_{Q}$ according to ET (Eq.~\ref{eq:equilibria}) with horizon $H$.
        \FOR{$t=0,1,\ldots,T-1$}
                        \STATE \lIf{$t=0$} {$\mu_{t,0}^0 = \mu_{0}$} \lElse {$\mu_{t,0}^0 = \mu_{t-1,1}^{K}$}
        \STATE Initialize $\mu_t^0\leftarrow\Gamma_{\mathcal M}(\pi^0)$ as the MF induced by $\pi^0$ with initial MF $\mu_{t,0}^0$.  
        \STATE Initialize weighted sum of policies $\bar \pi^0_t = \pi^0 \mu^0_t$.
            \FOR {$k=0,1,\ldots,K-1$} 
        \STATE Evaluate lines 6-10 of Alg. 2.
        \ENDFOR
        \ENDFOR
        \STATE \textbf{return} $\pi_t^K$ for $t=0,1,\ldots,T-1$.
    \end{algorithmic}
\end{algorithm}
\paragraph{Computation of RH equilibria.}
RH equilibria can be computed in a sequential manner by applying the GFPI or GFP algorithm for the arising MFGs at each time step. The result of each algorithm yields the initial condition for the next MFG, which starts with the previous MF after one time step. We summarize the sequential RH-GFP variant in Alg.~\ref{alg3}. 

\begin{theorem} \label{thm:gfp2}
    Under Assm.~\ref{ass:monotone} and \ref{ass:pnomf}, the continuous-time versions of sequential RH-GFP for RE converges to zero exploitability as $\tau \to \infty$.
\end{theorem}

This sequential approach, however, can be inefficient, especially for long horizons, since every MFG has to wait until the previous MFG converged.
To circumvent this, we propose a parallel algorithm, where in each iteration for each MFG we change the initial condition to the previous MF after one time step. This parallelization is efficient, since the later MFGs start learning equilibria before their initial condition has converged. We summarize the parallel RH-GFP in Alg.~\ref{alg4} in the appendix, where we also provide experiments which highlight the efficiency compared to the sequential approach.

\begin{figure}
    \centering
    \includegraphics{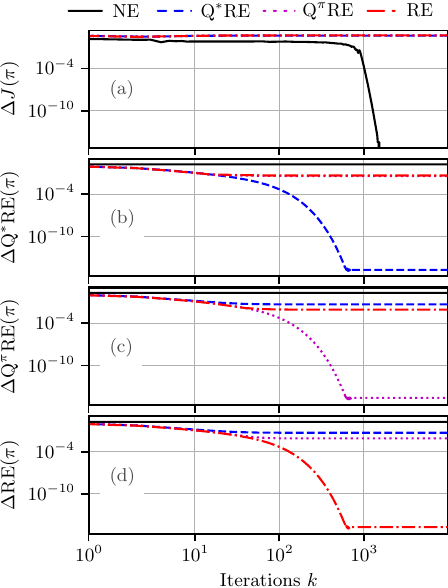}
    \caption{Convergence of GFP for the Susceptible-Infectious-Susceptible MFG with $\alpha = 1.0$ and $\beta=0.95$. The GFP algorithms for $\QRE$, $\BE$ and RE show similar behaviour in the first iterations before converging to their respective equilibria. }
    \label{fig:SIS_GFP}
\end{figure}
\section{Experiments} \label{sec:exp}
In this section we evaluate our algorithms and analyze the different equilibria for several MFGs.
We analyze the efficacy of our methods for a Susceptible-Infectious-Susceptible (SIS) problem and a sequential version of Rock-Paper-Scissor (RPS) game.
Additionally, we evaluate random MFGs, similar to \citet{perolat2021scaling}, by creating random transition and reward matrices and adding a mean-field dependent function to the reward that promotes swarm avoiding behaviour.
Detailed game descriptions are found in the appendix. For code, see https://github.com/yannickeich/QRE-MFG.

To measure algorithm efficiency, we quantify the distance to a $\QRE$ / $\BE$ / RE, similar to exploitability in the NE case.

\begin{definition}[Distance to equilibria]
    The distance of a policy $\pi \in \Pi$ to a $\QRE$ / $\BE$ / RE is defined as
    \begin{align*}
        \Delta \QRE(\pi) &\coloneqq \max_{t \in \mathcal T} \left\Vert \pi_t - \Gamma^{ 1/ \lambda}_{\Pi}(\Gamma_{Q^{\pi}}(\Gamma_{\mathcal M}(\pi), \pi))_t \right\Vert, \\
        \Delta \BE(\pi) &\coloneqq \max_{t \in \mathcal T} \left\Vert \pi_t - \Gamma^{1/ \lambda}_{\Pi}(\Gamma_{Q^*}(\Gamma_{\mathcal M}(\pi)))_t \right\Vert,\\
        \Delta \mathrm{RE}(\pi) &\coloneqq \max_{t \in \mathcal T} \left\Vert \pi_t - \Gamma^\alpha_{\Pi}(\Gamma^\alpha_{\tilde Q}(\Gamma_{\mathcal M}(\pi)))_t \right\Vert.
    \end{align*}
\end{definition}
Note that whenever $\Delta \QRE(\pi)$, $\Delta \BE(\pi)$ or $\Delta \mathrm{RE}(\pi)$ is zero for a policy $\pi$, $\pi$ is a $\QRE$ / $\BE$ / RE of the MFG.
\begin{figure}
    \centering
    \includegraphics{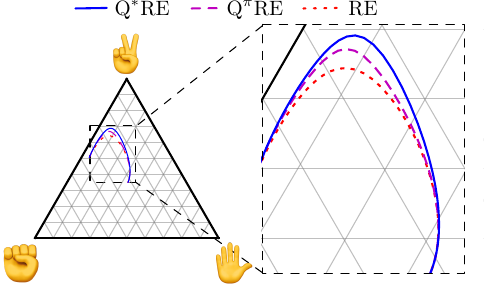}
    \caption{Action probabilities in the Rock-Paper-Scissor problem at $t=0$ for the resulting $\BE$ / $\QRE$ / RE  using GFP (Alg.~\ref{alg2}) over various temperatures $\alpha = 1/\lambda$. As $\alpha \to \infty$, we always obtain the uniform policy in the center, while as $\alpha \to 0$, solutions converge to the Nash solution (to the left). In-between, solutions differ from each other, regardless of the temperature.}
    \label{fig:simplex}
\end{figure} 

First, we employ the GFP algorithm to compute $\QRE$, $\BE$ and RE for the SIS MFG. Figure~\ref{fig:SIS_GFP} illustrates the progress of the algorithms by displaying $\Delta J(\pi),\Delta \QRE(\pi)$, $\Delta \BE(\pi)$ and $\Delta \mathrm{RE}(\pi)$ over the iterations $k$. The results display the efficacy of the GFP algorithm, showing fast convergence to the equilibria. The exploitability plot highlights the bounded rationality of the other equilibria compared to the NE. In exchange for converging to their respective equilibria, GFP for $\QRE$, $\BE$ and RE does not lead to zero $\Delta J(\pi)$. Additionally, Figure~\ref{fig:SIS_GFP} highlights both the similarities and differences of $\QRE$, $\BE$ and RE. Their respective algorithms behave similarly in the first iterations before leading to the distinct equilibria.

To further visualize the distinctiveness of the equilibrium concepts, we compute $\QRE$, $\BE$ and RE with different temperatures $\alpha=1/\lambda$ for the sequential RPS MFG using the GFP algorithm. We indicate the policies of the first time step of the resulting equilibria in the probability simplex in Figure~\ref{fig:simplex}. The comparison experimentally verifies our previous discussion of the connections between the equilibrium notions. The results highlight both the equivalence of the equilibrium concepts for $\alpha \to 0$ (representing the NE) and $\alpha \to \infty$ (representing the uniform policy) and their distictiveness for intermediate temperatures.

Finally, we examine how combining $\QRE$ with RH equilibria can effectively model agents with limited lookahead capacities.
To do this, we employ the parallel RH-GFP algorithm to compute RH QRE of a random MFG for different receding horizons $H$.
Figure~\ref{fig:random_RHQRE} illustrates the distance of the various RH QRE to the $\QRE$ with time horizon $\mathcal T$ over the iterations $k$.
The results demonstrate that RH QRE with higher lookahead are closer to the $\QRE$. Conversely, assuming shorter lookahead capacities results in equilibria that deviate more from the $\QRE$, highlighting the efficacy of modeling bounded rationality with receding horizon.
\begin{figure}
    \centering
    \includegraphics{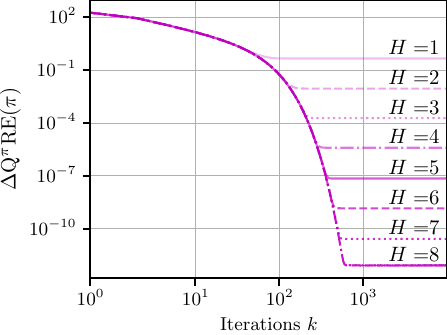}
   \caption{Comparison of the distance of various RH $\QRE$ with different horizons $H$ to the $\QRE$ with total horizon $\mathcal{T}$ for a random MFG with $\alpha=1.0$ and $\beta=0.95$ over iterations $k$. The equilibria induced by assuming shorter lookahead capacities deviate more from the QRE with total horizon, demonstrating the impact of limited lookahead on equilibrium behavior.}
    \label{fig:random_RHQRE}
\end{figure}

We have included additional experiments in the appendix that demonstrate the convergence of our proposed algorithms and provide further insights into the connections between the different equilibrium concepts.

\section{Conclusion}
In this work, we introduced both QRE MFGs and RH MFGs, to incorporate bounded rationality in MFGs for a more realistic modeling of large agent populations. We compared these new equilbrium concepts to existing ones theoretically and empirically. Our analysis highlights the similarities and differences of the discussed equilibrium concepts. Furthermore, we designed general learning algorithms to compute equilibria efficiently and evaluated these algorithms on different problem settings. We hope that the novel RH and QRE MFGs combined with our learning algorithms help bring the theory closer to real-world scenarios with 
not perfectly rational agents. For future work, one could apply our theory and learning approach to research problems where bounded rationality is crucial, e.g. in economics or the social sciences.

\section*{Acknowledgements}
This work has been co-funded by the LOEWE emergenCITY research promotion program of the federal state of Hessen, Germany, by the German Research Foundation (DFG) within the Collaborative Research Center (CRC) 1053 MAKI and project number 517777863, by the Federal Ministry of Education and Research as part of the Software Campus project RL4MFRP (funding code 01IS23067) and by the Hessian Ministry of Science and the Arts (HMWK) within the projects "The Third Wave of Artificial Intelligence - 3AI" and hessian.AI. 


\bibliography{references}

\newpage
\appendix
\onecolumn

\section{A \quad Proof of Proposition~\ref{prop:disc-qre-exist}}
\begin{proof}[Proof of Prop.~\ref{prop:disc-qre-exist}]
    We observe that the functions $\pi \mapsto Q^{\pi}$ and $Q^{\pi} \mapsto \frac{\exp \left(\lambda Q_t^{\pi}(x,u)\right) }{\sum_{u' \in \mathcal{U}}\exp \left(\lambda Q_t^{\pi}(x,u')\right) }$ are continuous for any $\lambda > 0$ due to finiteness of $\mathcal T, \mathcal X, \mathcal U$ as compositions of continuous functions by Assm.~\ref{ass:cont}. Therefore, the existence follows by Brouwer's fixed point theorem.
\end{proof}

\section{B \quad Proof of Proposition~\ref{prop:disc-conv-fpi}} \label{app:prop:disc-conv-fpi}
\begin{proof}[Proof of Prop.~\ref{prop:disc-conv-fpi}]
    For BE and RE, the result is known from \citep[Thm.~3]{cui2021approximately}. For QRE, the proof follows similarly by obtaining Lipschitz continuity of the fixed point map. In particular, the map $\Gamma_{\mathcal M}$ from policy to MF is Lipschitz by Assm.~\ref{ass:Lcont}. Similarly, the map $\pi \mapsto \Gamma_{Q}(\Gamma_{\mathcal M}(\pi), \pi)$ mapping to the state-action value function is Lipschitz. Lastly, the map $\Gamma^{\alpha}_\Pi$ is Lipschitz with a constant that is inversely proportional to $\alpha$ (\citet[Prop.~4]{gao2017properties} or \citet[Lemma~B.7.5]{cui2021approximately}). Therefore, for sufficiently large $\alpha > 0$, we have that $\pi \mapsto \Gamma^\alpha_{\Pi}(\Gamma_{Q}(\Gamma_{\mathcal M}(\pi), \pi))$ is a contraction.
\end{proof}

\section{C \quad Proof of Theorem~\ref{thm:gfp1}} \label{app:thm:gfp1}
For the proof, we define the FP process for the MF at all times $t$ and continuous algorithm iterations $\tau$ as
\begin{align}
\begin{split} \label{eq:fp1}
    \bar \mu^\tau_t &\coloneqq \frac 1 \tau \int_{0}^{\tau} \mu^{\pi^{*, \tau'}}_{t} \mathrm d\tau', \\
    \bar \pi^\tau_t &\propto \frac 1 \tau \int_{0}^{\tau} \mu^{\pi^{*, \tau'}}_{t} \pi^{*, \tau'}_t \mathrm d\tau',
\end{split}
\end{align}
where $\mu^{\pi^{*, \tau'}} = \Gamma_{\mathcal M}(\pi^{*, \tau'})$ is generated by the best regularized response against the current MF $\bar \mu^{\tau'}$, i.e. $\pi^{*, \tau'} = \Gamma^\alpha_{\Pi}(\Gamma^\alpha_{\tilde Q}(\bar \mu^{\tau'})$ for all $\tau' > 1$. The initial values for $\tau \leq 1$ are given by some arbitrary initial policy and its MF, as in \cite{perrin2020fictitious}.

As a result of Assm.~\ref{ass:pnomf}, we keep in mind that
\begin{align*}
    \mu^{\bar \pi^\tau}_{t} = \bar \mu^{\tau}_{t}
\end{align*}
at all times $t$, and hence further for all $t, x,  u$
\begin{align}
\label{eq:proof_property}
    \frac{\mathrm d}{\mathrm d \tau} \mu^{\bar \pi^\tau}_{t}(x) &= \frac{1}{\tau} \left( \mu^{ \pi^{*, \tau'}}_{t}(x) - \mu^{\bar \pi^\tau}_{t}(x) \right) \\
    \label{eq:proof_property2}
    \mu^{\bar \pi^\tau}_{t}(x) \frac{\mathrm d}{\mathrm d\tau} \bar \pi^\tau_h(u \mid x) &= \frac{1}{\tau} \mu^{ \pi^{*, \tau'}}_{t}(x) \left( \pi^{*, \tau}_{t}(u \mid x) - \bar \pi^\tau_h(u \mid x) \right).
\end{align}

Consider the regularized objective $J^{\mathrm{RE}}_\alpha(\pi', \bar \pi^\tau)$ when following $\pi' \in \Pi_H$ under the MF $\bar \mu^\tau$, given as
\begin{align*}
    J^{\mathrm{RE}}_\alpha(\pi', \bar \pi^\tau) = \sum_{t \in \mathcal T} \sum_{x \in \mathcal X} \mu^{\pi'}_{t}(x) \sum_{u \in \mathcal U} \pi'_{t}(u \mid x) \left( r(x, u, \bar \mu^\tau) + \alpha \mathcal H(\pi'_t(\cdot \mid x)) \right)
\end{align*}
and optimized by the best response $\pi^{*, \tau'} = \Gamma^\alpha_{\Pi}(\Gamma^\alpha_{\tilde Q}(\bar \mu^{\tau'})$. The regularized exploitability of $\bar \pi^\tau$ against itself is then defined as
\begin{align*}
    \mathcal{E}^{\mathrm{RE}}(\bar \pi^\tau) \coloneqq \max_{\pi'} J^{\mathrm{RE}}_\alpha(\pi', \bar \pi^\tau) - J^{\mathrm{RE}}_\alpha(\bar \pi^\tau, \bar \pi^\tau)
\end{align*}
and can be analyzed in the following to show the convergence to regularized equilibria.

\begin{proof}[Proof of GFP convergence to RE]
We start by taking the derivative with respect to the algorithm time $\tau$ of the exploitability,
\begin{align*}
    &\frac{\mathrm d}{\mathrm d \tau} \mathcal{E}^{\mathrm{RE}}(\bar \pi^\tau) = \frac{\mathrm d}{\mathrm d\tau} \left( \max_{\pi'} J^{\mathrm{RE}}_\alpha(\pi', \bar \pi^\tau) - J^{\mathrm{RE}}_\alpha(\bar \pi^\tau, \bar \pi^\tau) \right) \\
    &\quad = \sum_{t \in \mathcal T} \sum_{x \in \mathcal X} \mu^{\pi^{*, \tau}_{t}}_{t}(x) \sum_{u \in \mathcal U} \pi^{*, \tau}_{t}(u \mid x) \frac{\mathrm d}{\mathrm d\tau} r(x, u, \bar \mu^\tau) \\
    &\qquad - \sum_{t \in \mathcal T} \sum_{x \in \mathcal X} \frac{\mathrm d}{\mathrm d\tau} \mu^{\bar \pi^\tau}_{t}(x) \sum_{u \in \mathcal U} \bar \pi^\tau_t(u \mid x) \left( r(x, u, \bar \mu^\tau) + \alpha \log(\bar \pi^\tau_t(u \mid x)) \right) \\
    &\qquad - \sum_{t \in \mathcal T} \sum_{x \in \mathcal X} \mu^{\bar \pi^\tau}_{t}(x) \sum_{u \in \mathcal U} \frac{\mathrm d}{\mathrm d\tau} \bar \pi^\tau_t(u \mid x) \left( r(x, u, \bar \mu^\tau) + \alpha \log(\bar \pi^\tau_t(u \mid x)) \right) \\
    &\qquad - \sum_{t \in \mathcal T} \sum_{x \in \mathcal X} \mu^{\bar \pi^\tau}_{t}(x) \sum_{u \in \mathcal U} \bar \pi^\tau_t(u \mid x) \frac{\mathrm d}{\mathrm d\tau} r(x, u, \bar \mu^\tau) \\
    &\qquad - \alpha \sum_{t \in \mathcal T} \sum_{x \in \mathcal X} \mu^{\bar \pi^\tau}_{t}(x) \sum_{u \in \mathcal U} \bar \pi^\tau_t(u \mid x) \frac{\mathrm d}{\mathrm d\tau} \log(\bar \pi^\tau_t(u \mid x)) 
\end{align*}
where the latter equality follows by dropping derivatives w.r.t. $\pi^{*, \tau}$ via envelope theorem \cite{afriat1971theory}, since $\mathcal{E}^{\mathrm{RE}}(\bar \pi^\tau)$ is continuously differentiable in $\tau$ by continuity of the maps $\Gamma^\alpha_{\Pi}$, $\Gamma^\alpha_{\tilde Q}$ and hence also \eqref{eq:fp1}.

First, note that for the last term, we obtain
\begin{align*}
    &\alpha \sum_{t \in \mathcal T} \sum_{x \in \mathcal X} \mu^{\bar \pi^\tau}_{t}(x) \sum_{u \in \mathcal U} \bar \pi^\tau_t(u \mid x) \frac{\mathrm d}{\mathrm d\tau} \log(\bar \pi^\tau_t(u \mid x)) \\
    &\quad = \alpha \sum_{t \in \mathcal T} \sum_{x \in \mathcal X} \sum_{u \in \mathcal U} \mu^{\bar \pi^\tau}_{t}(x) \frac{\mathrm d}{\mathrm d\tau} \bar \pi^\tau_t(u \mid x) \\
    &\quad = \alpha \sum_{t \in \mathcal T} \sum_{x \in \mathcal X} \frac{1}{\tau} \mu^{ \pi^{*, \tau'}}_{t}(x) \left( \sum_{u \in \mathcal U} \pi^{*, \tau}_{t}(u \mid x) - \sum_{u \in \mathcal U} \bar \pi^\tau_t(u \mid x) \right) = 0,
\end{align*}
where the last equation follows from the fact that the sum of the probabilities assigned by a policy to all possible actions is equal to one.
Moreover, for the first and fourth term, w.l.o.g. let $r$ not depend on actions, then we have
\begin{align*}
    &\sum_{t \in \mathcal T} \sum_{x \in \mathcal X} \mu^{\pi^{*, \tau}_{t}}_{t}(x) \sum_{u \in \mathcal U} \pi^{*, \tau}_{t}(u \mid x) \frac{\mathrm d}{\mathrm d\tau} r(x, u, \bar \mu^\tau) \\
    &\qquad - \sum_{t \in \mathcal T} \sum_{x \in \mathcal X} \mu^{\bar \pi^\tau}_{t}(x) \sum_{u \in \mathcal U} \bar \pi^\tau_t(u \mid x) \frac{\mathrm d}{\mathrm d\tau} r(x, u, \bar \mu^\tau) \\
    &\quad = \sum_{t \in \mathcal T} \sum_{x \in \mathcal X} \mu^{\pi^{*, \tau}_{t}}_{t}(x) \sum_{u \in \mathcal U} \pi^{*, \tau}_{t}(u \mid x) \left\langle \nabla_{\bar \mu^\tau} r(x, \bar \mu^\tau), \frac{\mathrm d}{\mathrm d\tau} \bar \mu^\tau \right\rangle \\
    &\qquad - \sum_{t \in \mathcal T} \sum_{x \in \mathcal X} \mu^{\bar \pi^\tau}_{t}(x) \sum_{u \in \mathcal U} \bar \pi^\tau_t(u \mid x) \left\langle \nabla_{\bar \mu^\tau} r(x, \bar \mu^\tau), \frac{\mathrm d}{\mathrm d\tau} \bar \mu^\tau \right\rangle \\
    &\quad = \tau \sum_{t \in \mathcal T} \sum_{x \in \mathcal X} \frac{\mathrm d}{\mathrm d\tau} \mu^{ \pi^{*, \tau'}}_{t}(x) \left\langle \nabla_{\bar \mu^\tau} r(x, \bar \mu^\tau), \frac{\mathrm d}{\mathrm d\tau} \bar \mu^\tau \right\rangle \leq 0, 
\end{align*}
where we use the monotonicity Assumption \ref{ass:monotone}.
Lastly, for the remaining terms (third and fourth), we obtain back the exploitability $- \frac 1 \tau \mathcal{E}^{\mathrm{RE}}(\pi^\tau)$ by applying the properties in  Equations (\ref{eq:proof_property}) and (\ref{eq:proof_property2}). Hence, we have convergence because $\frac{\mathrm d}{\mathrm d \tau} \mathcal{E}^{\mathrm{RE}}(\pi^\tau) \leq - \frac 1 \tau \mathcal{E}^{\mathrm{RE}}(\pi^\tau)$ implies $\mathcal{E}^{\mathrm{RE}}(\pi^\tau) \leq \frac{\mathcal{E}^{\mathrm{RE}}(\pi^0)}{\tau} \to 0$ by Gronwall's inequality.
\end{proof}

\section{D \quad Definitions and Notations for RHE}
In this section we first repeat the defintion of RH NE before introducing RH QRE.

The previously defined objective functions $J$, e.g. corresponding to NE or RE, are defined as sums over the whole time horizon $\mathcal T$. In the receding horizon scenario, however, each agent at time $t \in \mathcal T$ plans ahead by considering only the next $H$ time steps beyond the current time $t$.

For a RH MFG, we define the respective RH NE as follows. Note that the RH concept extends to RE equilibria by changing the corresponding objective.
\begin{definition}[RH NE]
    For a horizon $H \in \mathbb N$, a RH equilibrium is a policy ensemble $(\pi^{*,H}_{t})_{t \in \mathcal T} \in \Pi^{T}$ such that for all $t \in \mathcal T$
    \begin{align*}
        \pi^{*, H}_{t} \in \argmax_{\pi \in \Pi} J^H_{t} (\pi, \pi^{*, H}_{t})
        \coloneqq \argmax_{\pi \in \Pi} \E \left[ \sum_{t' = t}^{\min (T, t+H)} r \left(x_{t'}, u_{t'}, \mu_{t, t'}^H \right) \right] \qquad \textrm{with} \quad \mu_{t}^H = \Gamma_{\mathcal M, t}^H \left(\pi^{*, H}_{t} \right) ,
    \end{align*}
    and $(J_{t}^H)_{t \in \mathcal T}$ is the corresponding RH objective ensemble. Since the observation may change after each time step $t$ formalized by $\mu_{t}^H$, each agent effectively only implements the first time step $\pi^{*, H}_{t, t}$ of each policy. Therefore, the actually implemented equilbrium policy $\pi^{**,H} \in \Pi$ results from the policy ensemble $(\pi^{*, H}_{t})_{t \in \mathcal T} \in \Pi^{T}$ by taking
    \begin{align*}
        \pi^{**, H}_{t} = \pi^{*, H}_{t, t},
    \end{align*}
    for each $t \in \mathcal{T}$ and can be seen as the diagonal $\pi^{**,H} = \mathrm{diag} ((\pi^{*, H}_{t})_{t \in \mathcal T})$ of ensemble $(\pi^{*, H}_{t})_{t \in \mathcal T}$.
\end{definition}
In the above equilibrium definition, each $\mu_{t}^H = \Gamma_{\mathcal M, t}^H (\pi) \in \mathcal{P} (\mathcal{X})^{\min (T, t +H)}$ with $t \in \mathcal T$ is a mean field ensemble with initial state distribution $\mu_{t, t}^H = \mu_{t-1, t}^H$ at time $t$ and recursively calculated
\begin{align*}
    \mu_{t, t'+1}^H (x') = \sum_{x \in \mathcal X} \mu_{t, t'}^H(x) \sum_{u \in \mathcal U} \pi_{t'}(u \mid x) p_{t'} \left(x' \mid x, u, \mu_{t, t'}^H \right)
\end{align*}
for each $t' \in \{ t, \ldots, \min (T, t+H) - 1 \}$.

\paragraph{RH QRE.} To define a RH QRE with receding horizon $H \in \mathbb N$, we first provide an ensemble of $Q$-function sequences $(\hat{Q}_t^{H, \pi})_{t \in \mathcal T}$ with $\hat{Q}_t^{H, \pi} \coloneqq (\hat{Q}_{t, t'}^{H, \pi})_{t' \in \{ t, \ldots, \min (T, t + H) \} }$ for each $t \in \mathcal T$. Here, we first define $Q_t^{H, \pi} \coloneqq (Q_{t, t'}^{H, \pi})_{t' \in \{ t, \ldots, \min (T, t + H) \} }$ for each $t \in \mathcal T$ through the recursion
\begin{align*}
    Q_{t, t'}^{H, \pi} (x,u) = r(x,u,\mu_{t'}) + \sum_{x' \in \mathcal{X}} p_{t'} (x'\mid x,u,\mu_{t'}) \sum_{u' \in \mathcal U} \pi_{t'}(u' \mid x') Q_{t, t'+1}^{H, \pi}(x', u')
\end{align*}
with terminal condition $Q_{t, \min (T, t + H) }^{H, \pi} \coloneqq 0$. Then, we define for all $t \in \mathcal T$ and $t' \in \{ t, \ldots, \min (T, t + H) \}$
\begin{align*}
    \hat Q^{H, \pi}_{t, t'}(x,u) \coloneqq Q^{H, \pi}_{t, t'}(x,u) + \epsilon_{t, t'}(x,u),
\end{align*}
where the $\epsilon_{t, t'} (x,u)$ are sampled i.i.d. from some distribution $p_\epsilon$. Analogous to the standard case, we define the RH response set
\begin{align*}
    \mathcal{R}_{t, t',x,u}^{H, \mu} = \{\epsilon_{t, t'} \colon \hat Q^{H, \pi, \mu}_{t, t'}(x,u) > \hat Q^{H, \pi, \mu}_{t, t'}(x,u') \quad \forall u'\neq u \}.
\end{align*}

\begin{definition}[RH QRE]
    A (Markov) RN QRE is a policy ensemble $(\pi^{*, H}_t)_{t \in \mathcal T}$ s.t. for each $\pi^{*, H}_t = (\pi^{*, H}_{t, t'})_{t' \in \{ t, \ldots, \min (T, t + H) \} }$
    \begin{equation*}
        \pi_{t, t'}^{*, H}(u\mid x) = \int_{\mathcal{R}^{H, \mu}_{t,t',x,u}}p(\epsilon) \mathrm d\epsilon
    \end{equation*}
    for all $t \in \mathcal T, t' \in \{ t, \ldots, \min (T, t + H) \}, x \in \mathcal X, u \in \mathcal U$, and $\mu = \Gamma_{\mathcal M, t}^H \left(\pi^{*, H}_{t} \right)$. As before, the actually implemented equilbrium policy $\pi^{**,H} \in \Pi$ results from the policy ensemble $(\pi^{*, H}_{t})_{t \in \mathcal T} \in \Pi^{T}$ by taking
    \begin{align*}
        \pi^{**, H}_{t} = \pi^{*, H}_{t, t},
    \end{align*}
    for each $t \in \mathcal{T}$ and can be seen as the diagonal $\pi^{**,H} = \mathrm{diag} ((\pi^{*, H}_{t})_{t \in \mathcal T})$ of ensemble $(\pi^{*, H}_{t})_{t \in \mathcal T}$.
\end{definition}

\section{E \quad Proof of Theorem~\ref{thm:gfp2}} \label{app:thm:gfp2}
In this section, we further prove the convergence of FP in the regularized receding-horizon case.

\paragraph{Proof for sequential RH-GFP.}
In the sequential version, the proof follows directly from Thm.~\ref{thm:gfp1}. In particular, for any $t \in \mathcal T$ in ascending order, we compute $\pi^{*, H}_{t}$ as 
\begin{align*}
    \bar \mu^{\tau, H}_{t,t'} &\coloneqq \frac 1 \tau \int_{0}^{\tau} \mu^{\pi^{*, \tau'}}_{t,t'} \mathrm d\tau', \\
    \bar \pi^{\tau,H}_{t,t'} &\propto \frac 1 \tau \int_{0}^{\tau} \mu^{\pi^{*, \tau'}}_{t,t'} \pi^{*, \tau'}_{t,t'} \mathrm d\tau',
\end{align*}
which converges to an $\epsilon$-optimal $\pi^{*, H}_{t}$ for $\tau > \tau_t$ at some $\tau_t > 0$ by Thm.~\ref{thm:gfp1}. Here, $\mu^{\pi^{*, \tau'}}_{t}$ solves by the usual forward equations but starts from $\mu^{\pi^{*, \tau'}}_{t-1, t}$, or $\mu_0$ for $t = 0$. This starting condition is also the one achieved by the previous computed policies at times before $t$. Therefore, the exploitability of the RH-GFP solution tends to zero at all times, i.e. $\mathcal{E}^{\mathrm{avg}} \to 0$ as $\tau \to \infty$. Choosing some $\tau^* > \max_{t \in \mathcal T}{\tau_t}$ for any fixed $\epsilon > 0$, we have an $\epsilon$-RH-RE $(\bar \pi^{\tau^*,H}_{t})_{t \in \mathcal T}$.

\newpage
\section{F \quad Parallel RH-GFP}
In this section we provide the pseudocode for the parallel computation of the RH equilibria.

\begin{algorithm}
    \caption{Parallel RH-GFP.}
    \label{alg4}
    \begin{algorithmic}[1]
        \STATE Input: Temperature $\alpha > 0$, policy $\pi^0$, $\beta \in (0, 1)$, equilibrium type ET $\in \{\text{NE}, \QRE, \BE, \text{RE}\}$.     
        \STATE Define $\Gamma_{\Pi}$ and $\Gamma_{Q}$ according to ET (Eq.~\ref{eq:equilibria}).
        \STATE Initialize $\mu_t^0\leftarrow\Gamma_{\mathcal M}(\pi^0)$  for $t=0,1,\ldots,T-1$
        \STATE Initialize $\bar \pi^0_t = \pi^0 \mu^0_i$ for $t=0,1,\ldots,T-1$
        \FOR {$k=0,1,\ldots,K-1$} 
            \FOR{$t=0,1,\ldots,T-1$}
        \STATE Evaluate lines 6-10 of Alg. 2.
        \STATE \lIf{$i=0$} {$\mu_{t,0}^k = \mu_{0}$} \lElse {$\mu_{t,0}^k = \mu_{t-1,1}^{k-1}$}
        \ENDFOR
        \ENDFOR
        \STATE \textbf{return} $\pi_t^K$ for $t=0,1,\ldots,T-1$.
    \end{algorithmic}
\end{algorithm}

\section{G \quad Benchmarks} 
In this section we describe the MFGs used for evaluation.
\subsection{Modified RPS Game}
The modified RPS Game consists of the  states $\mathcal{X} = \{\mathrm{Start},\mathrm{Rock},\mathrm{Paper},\mathrm{Scissor}\}$ and the actions $\mathcal{U} = \{\mathrm{Rock},\mathrm{Paper},\mathrm{Scissor}\}$.
At the initial time point the mean field starts in the Start state, i.e., $\mu_0(x=\mathrm{Start})=1.0$. The agents choose in each time point, in which state they want to jump, while the probability that this jump succeeds, depends on the current mean field. If the jump does not succeed, the agent stays in its state. The transition probabilities of a agent in state $\mathrm{Paper}$ choosing action $\mathrm{Rock}$ yield:
\begin{align*}
    p_t(x_{t+1} = \mathrm{Rock} \mid x_t = \mathrm{Paper}, u_t = \mathrm{Rock}, \mu_t) &= 1 - \mu_t(x=\mathrm{Rock})\\
        p_t(x_{t+1} = \mathrm{Paper} \mid x_t = \mathrm{Paper}, u_t = \mathrm{Rock}, \mu_t) &= \mu_t(x=\mathrm{Rock}).
\end{align*}

At each time step, an agent receives a positive reward based on the proportion of the mean-field it dominates in a classic RPS game, and conversely, incurs a negative reward proportional to the portion of the mean-field it is defeated by.
We modify the reward function with some weights to ensure that the NE is not the uniform policy:

\begin{align*}
    r_t(x_t=\mathrm{Rock},\cdot) &= - a \mu_t({x=\mathrm{Paper}}) + b \mu_t(x=\mathrm{Scissor})\\
      r_t(x_t=\mathrm{Paper},\cdot) &= - c \mu_t({x=\mathrm{Scissor}}) + a \mu_t(x=\mathrm{Rock})\\
        r_t(x_t=\mathrm{Scissor},\cdot) &= - b \mu_t({x=\mathrm{Rock}}) + c \mu_t(x=\mathrm{Paper}).
\end{align*}
We choose $a=10.0, b=1.0$ and $c=10.0$.

\subsection{Random MFGs}
The random MFGs that were used for evaluation in the experiment section consist of $n_x = 100$ states and $n_a = 10$ actions and were evaluated for a time horizon $\mathcal T = 10$. The $\mathcal T \times n_x \times n_x \times n_a$ sized transition probabilities and the $\mathcal T \times n_x \times n_a$ reward table where created using a random value generator. We then add the mean field dependent function $\bar r(x,\mu) = -\eta \log (\mu(x))$ to the reward that promotes social distancing, with $\eta = 1.0$.

\subsection{SIS Game }
The SIS game consists of the states $\mathcal{X}=\{\mathrm S,\mathrm I\}$ for susceptible and infectious and the actions $\mathcal U =\{\mathrm N, \mathrm Q\}$ for no quarantine and quarantine.
The dynamics of the MFG are defined by the following transition probabilities,

\begin{align*}
    p_t(x_{t+1} = \mathrm S \mid x_t = \mathrm I, u_t = \mathrm N, \mu_t) &= \gamma\\
    p_t(x_{t+1} = \mathrm S \mid x_t = \mathrm I, u_t = \mathrm Q, \mu_t) &= \gamma\\
    p_t( x_{t+1} = \mathrm{I} \mid x_t = \mathrm{S}, u_t = \mathrm{N}, \mu_t) &= \kappa \mu(x=\mathrm{I})  \\
    p_t(x_{t+1} = \mathrm{I} \mid x_t = \mathrm{S}, u_t = \mathrm{Q}, \mu_t) &= 0.0,
\end{align*}
where $\gamma = 0.4$ is the healing rate and $\kappa=0.81$ is the infection rate.
The running rewards are given by
\begin{align*}
    r_t(x_t=\mathrm{I},u_t=\mathrm{N},\mu_t) & = - c_i \\
    r_t(x_t=\mathrm{I},u_t=\mathrm{Q},\mu_t) & = - c_i -c_q\\
    r_t(x_t=\mathrm{S},u_t=\mathrm{N},\mu_t) & = 0.0 \\
    r_t(x_t=\mathrm{S},u_t=\mathrm{Q},\mu_t) & = - c_i -c_q,
\end{align*}
where $c_i=1.0$ is the cost for being infected and $c_q=0.5$ is the cost for staying in quarantine. 
During our experiment the initial mean field was $\mu_0(x=\mathrm{I}) = 1-\mu_0(x=\mathrm{S}) = 0.1 $.

\newpage
\section{H \quad Additional Experiments}
In this section we describe additional experiments to further demonstrate the convergence of our proposed algorithms. We run all experiments on a Macbook Pro with M1 chip.
While in the main text we mainly analyze the GFP algorithm, here we provide additional experiments with the GFPI algorithm. 
We employ the GFPI algorithm to compute the $\QRE$, $\BE$ and RE for a random MFG, the SIS problem and the sequential RPS game with $\alpha=1.0$. The computation of the NE is carried out using the GFP algorithm, as the GFPI algorithm fails to converge in this case, as discussed in the section about learning non-cooperative MF equilibria. Figures~\ref{fig:random_GFPI},~\ref{fig:SIS_GFPI} and \ref{fig:RPS_GFPI} show the progress of the algorithms by displaying $\Delta J(\pi),\Delta \QRE(\pi)$, $\Delta \BE(\pi)$ and $\Delta \mathrm{RE}(\pi)$ over the iterations $k$ for the random MFG, the SIS game and the RPS game, respectively. The results display the efficacy of the GFPI algorithm, showing fast convergence to the equilibria for the random MFG and the SIS problem. For the sequential RPS Game however, the GFPI algorithm does not converge. This highlights the need for a more general algorithm like the GFP algorithm.

As a comparison, we run the GFP algorithm on the same tasks with $\beta=0.95$. Figures~\ref{fig:random_GFP} and \ref{fig:RPS_GFP} show the results for the random MFG and the RPS game, respectively, while the results for the SIS game where shown in Figure \ref{fig:SIS_GFP} in the main text. We can see the convergence for all tasks. Due to the averaging of the MFs, the GFP is slower than the GFPI (when it converges), but converges more robustly. 

\begin{figure}[htbp]
    \centering
    \begin{minipage}{0.45\textwidth}
        \centering
        \includegraphics[width=\textwidth]{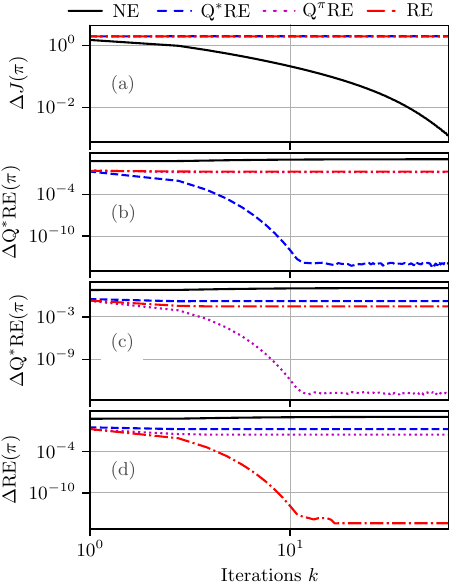}
        \caption{Convergence of GFPI for a random MFG with $\alpha = 1.0$.}
        \label{fig:random_GFPI}
    \end{minipage}
    \hfill
    \begin{minipage}{0.45\textwidth}
        \centering
        \includegraphics[width=\textwidth]{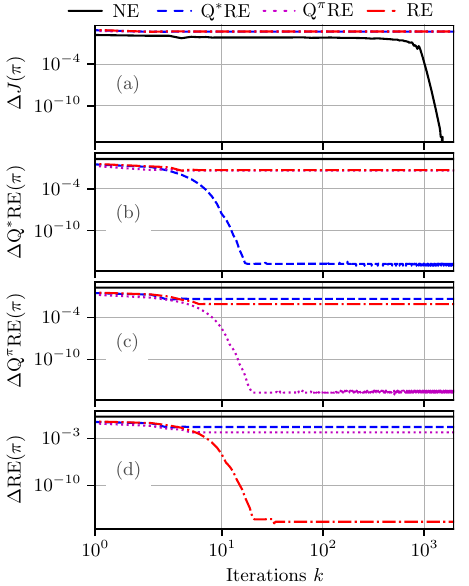}
        \caption{Convergence of GFPI for the SIS MFG with $\alpha = 1.0$.}
        \label{fig:SIS_GFPI}
    \end{minipage}
\end{figure}

\begin{figure}
    \centering
    \includegraphics{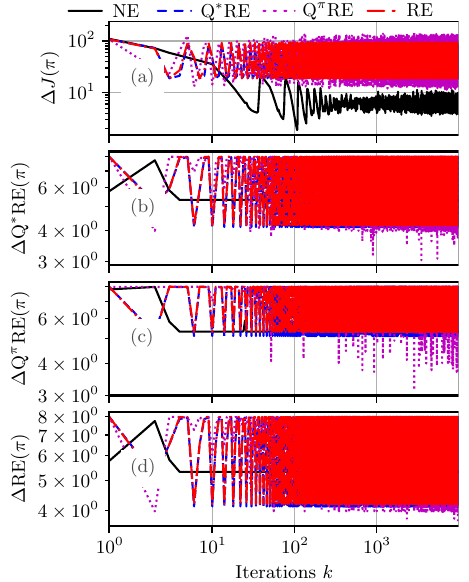}
    \caption{Nonvergence of GFPI for the RPS MFG with $\alpha = 1.0$.}
    \label{fig:RPS_GFPI}
\end{figure}

\begin{figure}[htbp]
    \centering
    \begin{minipage}{0.45\textwidth}
        \centering
        \includegraphics[width=\textwidth]{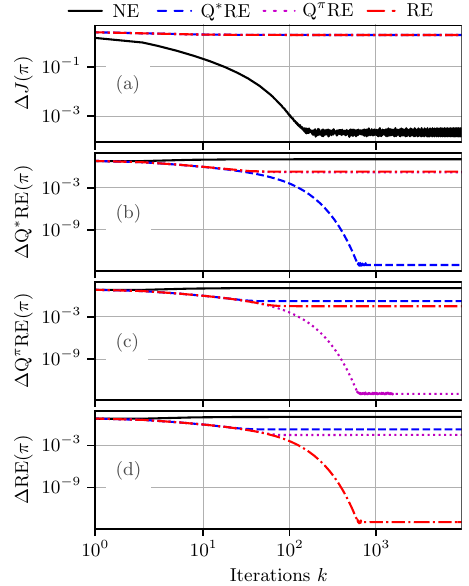}
        \caption{Convergence of GFP for a random MFG with $\alpha = 1.0$ and $\beta=0.95$.}
        \label{fig:random_GFP}
    \end{minipage}
    \hfill
    \begin{minipage}{0.45\textwidth}
        \centering
        \includegraphics[width=\textwidth]{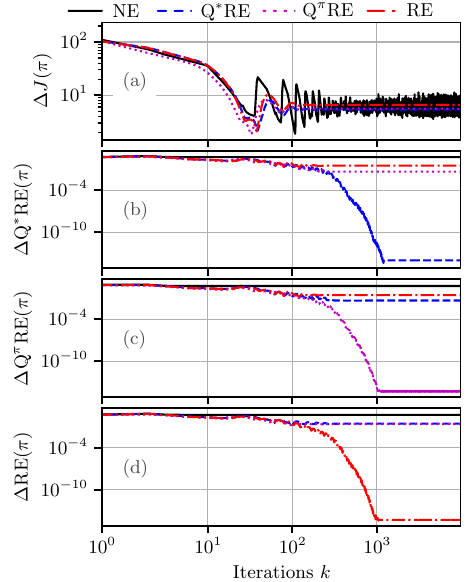}
        \caption{Convergence of GFP for the RPS MFG with $\alpha = 1.0$ and $\beta=0.95$.}
        \label{fig:RPS_GFP}
    \end{minipage}
\end{figure}

Next, we compare the parallel to the sequential RH GFP. Therefore we employ both algorithms on the RPS problem with $\mathcal{T}=10$ and $H=5$ to compute the RH $\QRE$. This leads to $\mathcal T - H +1 = 6$ underlying MFGs. Note that in our algorithms 3 and 4, we have $\mathcal T$ MFGs, but we can reduce it to $\mathcal T - H +1$ by applying all actions from the last equilibrium policy. Figure \ref{fig:parallel_vs_seq} shows the progress of both algorithms over the iterations $k$ for all $6$ MFGs. We see that the sequential algorithm takes approximately the same number of iterations $k=1000$ for all MFGs to converge. The parallel algorithm is equivalent to the sequential algorithm for the first MFG, resulting in the same convergence. For the following MFGs the number of iterations increase only little per game. As a result, the parallel algorithm is more efficient, needing significantly fewer iterations overall compared to the sequential algorithm, which requires 1000 iterations per MFG, while the parallel algorithm only needs a little more than 1000 iterations to compute all MFGs in parallel.

\begin{figure}
    \centering
    \includegraphics{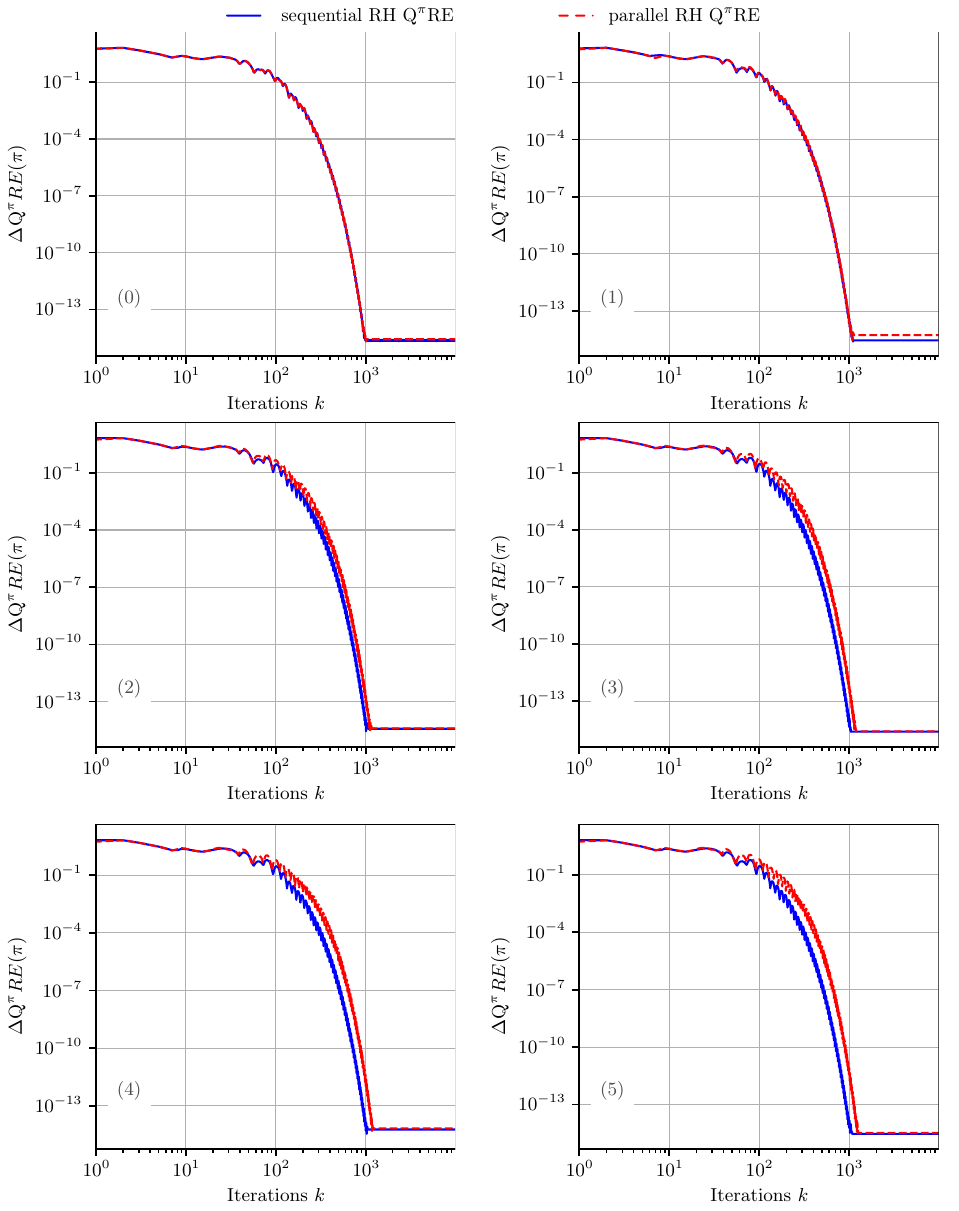}
    \caption{Comparion of the sequential and parallel RH GFP for the RPS MFG with $\alpha=1.0$, $\beta=0.95$, total horizon $\mathcal T = 10$ and receding horizon $H$. While the sequential algorithm requires approximately 1000 iterations per MFG, the parallel algorithm only needs a little more than 1000 iterations to compute all MFGs in parallel.}
    \label{fig:parallel_vs_seq}
\end{figure}

\end{document}